\documentclass[11pt]{article}
\usepackage{amsfonts}
\usepackage[margin=2cm]{geometry}
\usepackage{graphicx,color}
\usepackage[hang,small,bf]{caption}
\usepackage{xspace}
\usepackage{enumerate}
\usepackage{amssymb,amsmath}
\usepackage{amsthm}
\usepackage{tkz-graph}

\usepackage[T1]{fontenc}
\usepackage[utf8]{inputenc}
\usepackage{authblk}

\newtheorem{thm}{Theorem}

\newtheorem{corollary}[thm]{Corollary}
\newtheorem{theorem}[thm]{Theorem}

\newtheorem{lemma}[thm]{Lemma}

\newtheorem{property}[thm]{Property}

\newtheorem{definition}[thm]{Definition}

\newcommand\Voisin{\Gamma}

\newcommand\B[2]{b{(#1,#2)}}
\newcommand\Rr[2]{r_1{(#1,#2)}}
\newcommand\Rl[2]{r_2{(#1,#2)}}

\renewcommand\bot{{null}}
\newcommand\carre{\scriptscriptstyle\blacksquare}
\newcommand\rond{\displaystyle\bullet}

\title{The Manne \emph{et al.} self-stabilizing $\frac23-$approximation matching algorithm is sub-exponential.}

\author[1]{Johanne Cohen}
\author[2]{Jonas Lef\`evre}

\author[3]{Khaled Ma\^amra}
\author[1]{George Manoussakis}
\author[3]{Laurence Pilard}

\affil[1]{LRI-CNRS, Universit\'e Paris-Sud, Universit\'e Paris Saclay, France\\
\texttt{\{johanne.cohen,george.manoussakis\}@lri.fr}}
\affil[2]{ IRIF, Universit\'e Paris-Diderot -- Paris 7 , France, 
\texttt{\{jonas.lefevre\}@irif.fr}}

\affil[3]{LI-PaRAD, Universit\'e Versailles-St. Quentin, Universit\'e Paris Saclay,  France\\

\texttt{\{khaled.maamra, laurence.pilard\}@uvsq.fr}}

\begin{document}

\date{}
\maketitle
\thispagestyle{empty}

 \begin{abstract}
Manne \emph{et al.}  \cite{ManneMPT11} designed  the first algorithm  computing a maximal matching that is a $\frac{2}{3}$-approximation of the maximum matching in $2^{O(n)}$ moves.  However, the complexity tightness was not proved. In this paper, we exhibit a sub-exponential execution of this matching algorithm :  this algorithm  can stabilize after at most $\Omega(2^{\sqrt n})$ moves under the central daemon.
\end{abstract}


\section{Introduction}

In graph theory, a \emph{matching} $M$ in a graph is a set of edges without common vertices.  A matching is \emph{maximal} if no proper superset of $M$ is also a matching. A \emph{maximum} matching is a maximal matching with the highest cardinality among all possible maximal matchings. In this paper, we present a self-stabilizing algorithm for finding a maximal matching. Self-stabilizing algorithms \cite{Dijkstra74,Dolev00}, are distributed algorithms that recover after any transient failure without external intervention \emph{i.e.} starting from any arbitrary initial state, the system eventually converges to a correct behavior. The environment of self-stabilizing algorithms is modeled by the notion of \emph{daemon}. A daemon allows to capture the different behaviors of such algorithms accordingly to the execution environment. Three major types of daemons exist: the \emph{sequential}, the \emph{synchronous} and the \emph{distributed} ones. The sequential daemon means that exactly one eligible process is scheduled for execution at a time. The synchronous daemon means that every eligible process is scheduled for execution at a time. The distributed daemon means that any subset of eligible processes is scheduled for execution at a time. In an orthogonal way, a daemon can be \emph{fair} (meaning that every eligible process is eventually scheduled for execution) or \emph{adversarial} (meaning that the daemon only guarantees global progress, \emph{i.e.} at any time, at least one eligible process is scheduled for execution).

\section{Related Works}

Matching problems have received a lot of attention in different areas. Dynamic load balancing and job scheduling in parallel and distributed networks can be solved by algorithms using a matching set of communication links~\cite{BerenbrinkFM08,GhoshM96}. Moreover, the matching  problem has been recently  studied in the algorithmic game theory.  Indeed, the seminal problem relative to matching introduced by Knuth  is the stable marriage problem \cite{Knuth}. This problem can be modeled as a game with economic interactions such as two-sided markets \cite{AckermannGMRV11} or as a game with preference relations in a social network \cite{Hoefer13}.  

Several self-stabilizing algorithms have been proposed to compute maximal matching in unweighted or weighted general graphs.  For an unweighted graph, Hsu and Huang \cite{HsuH92} gave the first algorithm and proved a bound of $O(n^3)$ on the number of moves under a sequential adversarial daemon. The complexity analysis is completed by Hedetniemi et al. \cite{HedetniemiJS01} to $O(m)$ moves. Manne et al. \cite{ManneMPT11} presented a self-stabilizing algorithm for finding a $2/3$-approximation of a maximum matching.  The complexity of this algorithm  is proved to be $O(2^n)$ moves under a distributed adversarial daemon.

\section{Model}

A system consists of a set of processes where two adjacent processes can communicate with each other. The communication relation is typically represented by a graph $G = (V, E)$ where $|V | = n$ and $|E| = m$. Each process corresponds to a node in $V$ and two processes $u$ and $v$ are adjacent if and only if $(u,v)\in E$. The set of neighbors of a process $v$ is denoted by $\Voisin(v)$ and is the set of all processes adjacent to $v$. 

We consider one communication model : the \emph{state model}.  In the \emph{state model}, each process maintains a set of \emph{local variables} that makes up the \emph{local state} of the process.  A process can read its local variables and the local variables of its neighbors, but it can write only in its own local variables. A \emph{configuration} $C$ is a set of the local states of all processes in the system.  Each process executes the same algorithm that consists of a set of \emph{rules}. Each rule is of the form of $<guard> \to <command>$. The \emph{guard} is a boolean function over the variables of both the process and its neighbors. The \emph{command} is a sequence of actions assigning new values to the local variables of the process.

A rule is \emph{enabled} in a configuration $C$ if the guard is true in $C$. A process is \emph{activable} in a configuration $C$ if at least one of its rules is enabled.  An \emph{execution} is an alternate sequence of configurations and transitions ${\cal E} = C_0,A_0,\ldots,C_i,A_i, \ldots$, such that $\forall i\in \mathbb{N}^*$, $C_{i+1}$ is obtained by executing the command of at least one rule that is enabled in $C_i$ (a process that executes such a rule makes a \emph{move}). More precisely, $A_i$ is the non empty set of enabled rules in $C_i$ that has been executed to reach $C_{i+1}$ such that each process has at most one of its rules in $A_i$.  An \emph{atomic operation} is such that no change can takes place during its run, we usually assume an atomic operation is instantaneous. In the case of the state model, such an operation corresponds to a rule.  We use the following notation : $C_i \to C_{i+1}$.  An execution is \emph{maximal} if it is infinite, or it is finite and no process is activable in the last configuration. All algorithm executions considered in this paper are assumed to be maximal.

A \emph{daemon} is a predicate on the executions. We consider only the most powerful one: the \emph{distributed daemon} that allows all executions described in the previous paragraph.

An algorithm is \emph{self-stabilizing} for a given specification, if there exists a sub-set $\cal L$ of the set of all configurations such that : every execution starting from a configuration of $\cal L$ verifies the specification  (\emph{correctness})  and  starting from any configuration, every execution reaches a configuration of $\cal L$   (\emph{convergence}). $\cal L$ is called the set of \emph{legitimate configurations}.  
A configuration is \emph{stable} if no process is activable in the configuration. 
Both algorithms presented here, are \emph{silent}, meaning that once the algorithm stabilized, no process is activable. In other words, all executions of a silent algorithm are finite and end in a stable configuration.
Note the difference with a non silent self-stabilizing algorithm that has at least one infinite execution with a suffix only containing legitimate configurations, but not stable ones. 

We consider the following matching algorithm given by Manne \emph{et  al.}  \cite{ManneMPT11}.  This algorithm, denoted $\mathcal{M}^+$, computes a maximal matching that is a $\frac{2}{3}$-approximation of the maximum matching in $2^{O(n)}$ moves.  However, the complexity tightness was not proved. In this paper, we exhibit a sub-exponential execution of this matching algorithm.

\section{Algorithm $\mathcal{M}^+$ given by Manne \emph{et al.}  \cite{ManneMPT11}}

The algorithm $\mathcal{M}^{+}$ operates on an undirected graph $G=(V,E)$, where every node $v \in V$ has a unique identifier.  $\mathcal{M}^{+}$ assumes that there exists an underlying maximal matching algorithm, which has reached a stable configuration where a stable maximal matching $M$ has been built. Based on $M$, $\mathcal{M}^{+}$ builds a $\frac{2}{3}-$approximation of the maximum matching. To perform that, nodes search for augmenting paths of length three. 

An augmenting path is a path in the graph, starting and ending in an unmatched node, and where every other edge is either unmatched or matched; \emph{i.e.} for each consecutive pair of edges, exactly one of them must belong to the matching Let us consider the example in Figure \ref{fig:augm-path}.(a). In this figure, $v$ and $u$ are matched nodes and $x, y$ are unmatched nodes. The path $(y, u, v, x)$ is a 3-augmenting path.

Once an augmenting path is detected, nodes rearrange the matching accordingly, \emph{i.e.} transform  this path with one matched edge into a path with two matched edges  (see Figure~\ref{fig:augm-path}.(b)). 
This transformation leads to the deletion  of the augmenting path and increases by one the cardinality of the matching. The algorithm will stabilize when there are no augmenting paths of length three left.
Thus  the hypothesis of Karps's theorem \cite{Karp} eventually holds, giving a $\frac{2}{3}-$approximation of the maximum matching.

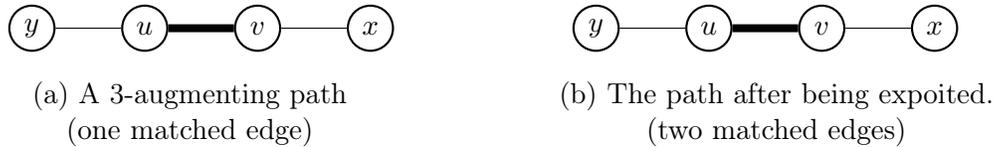
\begin{figure}[!hbt]
 \centering
\begin{tikzpicture}[scale=0.6]
	\tikzstyle{legend}=[rectangle, thin,minimum width=2.5cm, minimum height=1.2cm,black]
	\tikzstyle{node}=[draw,circle,minimum size=0.6cm,thick,inner sep=0pt]
	\tikzstyle{matched edge} = [draw,line width=3pt,-,black]
	\tikzstyle{edge} = [draw,-,black]
	\tikzstyle{pointer} = [draw, thick,->,>=stealth,black,shorten >=1pt]
	    \node[node] (v0) at  (0,0)  {$y$};
	    \node[node] (v1) at  (2.5,0)  {$u$};
	    \node[node] (v2) at  (5,0) {$v$};
	    \node[node] (v3) at  (7.5,0) {$x$};
	    \draw[matched edge] (v1) -- (v2);
	    \draw[edge] (v0) -- (v1);
	    \draw[edge] (v2) -- (v3);

	    \node[node] (v5) at  (0+12.5,0)  {$y$};
	    \node[node] (v6) at  (2.5+12.5,0)  {$u$};
	    \node[node] (v7) at  (5+12.5,0) {$v$};
	    \node[node] (v8) at  (7.5+12.5,0) {$x$};
	    \draw[matched edge] (v6) -- (v7);
	    \draw[edge] (v5) -- (v6);
	    \draw[edge] (v7) -- (v8);

\node[legend] (l1) at (3.5,-1.5) { (a) A 3-augmenting path};
\node[legend] (l1) at (3.5,-2.3) { (one matched edge)};

\node[legend] (l2) at (3.5+13,-1.5) { (b) The path after being expoited.};
\node[legend] (l2) at (3.5+13,-2.3) { (two matched edges)};

\end{tikzpicture}
 \caption{How to exploit a 3-augmenting path ?}
\label{fig:augm-path}
 \end{figure}

The underlying stable maximal matching $M$ is locally expressed by variables $m_v$ for each node $v$. These variables are defined as follows:\\
\indent $ \forall v \in V : (m_v=null) \Leftrightarrow( \forall (a,b) \in M,  a\neq v \land b\neq v )$ -- In this case, $v$ is called a \emph{single node} and we note $v\in \sigma(V)$.\\
\indent $ \forall v \in V : (\exists u \in V,\  m_v = u ) \Leftrightarrow( (v,u) \in M)$ -- In this case, $v$ is called a \emph{matched node} and we note $v\in \mu(V)$.\\

In Algorithm $\mathcal{M}^{+}$,  node $v$ keeps track of four variables, the pointer $p_v$ is used to define the final matching.  The variables $\alpha_v, \beta_v$ are used  to detect augmenting path and contains neighbors  of  $v$ that are single. Also, $s_v$ is a boolean variable  used for the augmenting path transformation.

Thus two neighboring nodes $v, u$ are matched in the final stable solution if and only if either 
$(p_v =u \land p_u =v)$ or if $(p_v =null \land p_u =null \land m_v =u \land m_u =v)$.\\

\noindent For each edge $(v,u)$ in $M$, matched nodes $v$ and $u$ are going to:
\begin{enumerate}
\item \textbf{Detect} augmenting path:  first, every pair of matched nodes $v,u$  will try to find single neighbors to which they can rematch. These single neighbors have to be \emph{available}, meaning they should not be involved in another augmenting path exploitation, \emph{i.e.} a single node $x$ is available if $p_x = null$. We will say that $x$ is a \emph{candidate} for $v$ if $x$ is an available single neighbor of $v$. Moreover $v$ and $u$ have to have a sufficient number of candidates to detect a 3-augmenting path: each node should have at least one candidate and the sum of the number of candidates for $v$ and $u$ should be at least 2. The \emph{BestRematch} predicate is used to compute  candidates in variables $\alpha$ and $\beta$, and the condition below (in \emph{AskFirst} predicate) is used to ensure the number of candidates is sufficiently high. 
($Unique(A)$ returns the number of unique elements in the multi-set $A$).
$$\alpha_u \neq null \wedge \alpha_v \neq null \wedge 2 \leq Unique( \{\alpha_u, \beta_u, \alpha_v, \beta_v\} ) \leq 4$$ 
\item \textbf{Try to exploit} this augmenting path : 
\begin{enumerate}
\item The \emph{AskFirst} node starts: exactly one of $v$ and $u$ will attempt to match with one of its candidates.
\item The \emph{AskSecond} node continues:  only when the first node succeeds will the second node also attempt to match with one of its candidates. 
\begin{enumerate}
\item If this also succeeds, the rematching is considered \emph{complete}. 
\item Otherwise the rematch built by the \emph{AskFirst} node is deleted and candidates $\alpha$ and $\beta$ are computed again, allowing then the detection of new augmenting paths. 
\end{enumerate}
\end{enumerate}
\end{enumerate}

\begin{figure}
{\small
\begin{tabbing}
aaa\=aaa\=aaa\=aaaaaaaa\=\kill
\noindent \textbf{SingleNode} \\
\> \textbf{if} $(p_v = null \wedge \mathit{Lowest} \{u\in \Voisin(v)~|~p_u = v\} \ne null) \vee p_v \notin \mu(\Voisin(v)) \cup \{null\} \vee $\\
\> \> $(p_v \ne null \wedge p_{p_v} \ne v)$\\
\> \textbf{then} $p_v := \mathit{Lowest} \{u\in \Voisin(v)~|~p_u = v\}$\\~\\

\>\>\>\>Algorithm $\mathcal{M}^+$ - Rule for nodes in $\sigma(V)$.\\~\\

\noindent \textbf{Update}\\
\> \textbf{if} $p_v \notin \sigma(\Voisin(v)) \cup \{null\}$ $\vee$ \\
\> \> $((\alpha_v, \beta_v) \ne$
\textbf{BestRematch$(v)$}$\,\wedge\,(p_v=null\,\vee\,p_{p_v} \notin \{v, null\}))$\\
\> \textbf{then} $(\alpha_v, \beta_v) :=$ \textbf{BestRematch$(v)$}\\
\> \> \hspace{0.6em} $(p_v, s_v) := (null, false)$\\~\\

\noindent \textbf{MatchFirst}\\
\textbf{Let} $x=$  \textbf{AskFirst}$(v,m_v)$ \\
\> \textbf{if} $x \ne null \wedge (p_v \ne x \vee
s_v \ne (p_{p_v} = v))$ \\  
\> \textbf{then} $p_v := x$\\
\> \> \hspace{0.6em} $s_v := (p_{p_v} = v)$\\~\\

\noindent \textbf{MatchSecond}\\
\textbf{Let} $y=$  \textbf{AskSecond}$(v, m_v)$ \\
\> \textbf{if} $y \ne null \wedge s_{m_v} = true \wedge p_v \ne y$ \\
\> \textbf{then} $p_v  := y$\\~\\

\noindent \textbf{ResetMatch}\\
\> \textbf{if} \textbf{AskFirst}$(v,m_v) =$ \textbf{AskSecond}$(v,m_v) = null \wedge (p_v, s_v) \ne (null, false)$\\
\> \textbf{then} $(p_v, s_v) := (null, false)$\\~\\

\>\>\>\>Algorithm $\mathcal{M}^+$ - Rules for nodes in $\mu(V)$.\\~\\

\noindent \textbf{BestRematch$(v)$}\\
\> $a := Lowest~ \{u\in \sigma(\Voisin(v)) \wedge (p_u = null \vee p_u = v)\}$\\
\> $b := Lowest ~\{u \in \sigma(\Voisin(v))\setminus\{a\} \wedge (p_u = null \vee p_u = v)\}$\\
\> \textbf{return} $(a, b)$\\~\\

\noindent \textbf{AskFirst}$(v,u)$\\
\> \textbf{if} $\alpha_v \neq null \wedge \alpha_u \neq null \wedge 2 \leq Unique( \{\alpha_v, \beta_v, \alpha_u, \beta_u\} ) \leq 4$ \\
\> \> \textbf{if} $ \alpha_v < \alpha_u \vee ( \alpha_v = \alpha_u \wedge \beta_v = null ) \vee ( \alpha_v = \alpha_u \wedge \beta_u \neq null \wedge v < u )$\\
\> \> \> \textbf{return} $\alpha_v$\\
\> \textbf{return} $null$\\~\\

\noindent \textbf{AskSecond}$(v,u)$\\
\> \textbf{if} $AskFirst(u,v) \neq null$\\
\> \> \textbf{return} $Lowest( \{\alpha_v, \beta_v \} \setminus \{\alpha_u \} )$ \\
\> \textbf{return} $null$\\~\\

\>\>\>\>Algorithm $\mathcal{M}^+$ - Functions
\end{tabbing}
}
\end{figure}

Now, we give a possible execution of Algorithm $\mathcal{M}^+$ under a distributed adversarial daemon.  Fig. \ref{fig:execution}.(a) shows the initial state of the execution. The topology is a path of seven vertices and the  identifiers of the nodes are indicated below. The underlying maximal matching represented by bold edges contains two edges $(24,2)$ and $(9,8)$. Then nodes $24$, $2$, $9$ and $8$ are \emph{matched} nodes (in $\mu(V)$) and nodes $15$, $10$ and $7$ are \emph{single} nodes (in $\sigma(V)$).  We illustrate the use of the \emph{$p$-values} by an arrow and the absence of the arrow means that the \emph{$p$-value} of the node equals to $null$.

At the beginning, there are two augmenting paths. Nodes $9$ and $8$ have already started to exploit their augmenting path. We are going to exhibit an execution where this augmenting path will be reset while the other one will be fully exploited.

In the initial configuration, we assume that all \emph{$\alpha$-values} and \emph{$\beta$-values} are defined as follows: 
$(\alpha_8,\beta_8)=(7,null)$, $(\alpha_9,\beta_9)=(10,null)$ and $(\alpha_{24},\beta_{24})  = (\alpha_2,\beta_2)=(null,null)$.  We also assume all \emph{$s$-values} are well defined: $s_8=true$ and $s_9=s_2=s_{24}=false$. At this step, node $9$ waits for an answer of node $10$. Nodes $2$ and $24$ have two unique candidates for a rematching.

At the beginning of the execution, all \emph{$\alpha$-values} and \emph{$\beta$-values} for all nodes are well defined except for nodes   $2$ and $24$ because \emph{BestRematch$(2)=(10, null)$},  \emph{BestRematch$(24)=(15, null)$}. Nodes $2$ and $24$ 
 execute a \emph{Update} move.  After these moves,  $(\alpha_{24},\beta_{24})=(15,null)$ and $(\alpha_2,\beta_2)=(10,null)$.  

Since $2 \leq Unique( \{\alpha_{2}, \beta_{2}, \alpha_{24}, \beta_{24}\} ) \leq 4$,  nodes $2$ and $24$ detect a $3$-augmenting path and start to exploit this augmenting path.  Since $AskFirst(2,24) = 10$ (which implies $AskFirst(24,2) = null$), node $2$ may execute a \emph{MatchFirst} move. Let us assume it does and then it points to node $10$, as seen in Figure~\ref{fig:execution}.(b). Since both nodes $9$ and $2$ are pointing to node $10$, node $10$ can choose the node to match with from these two nodes. Note that at this point, node $10$ is the only enabled node. Figure~\ref{fig:execution}.(c) shows the configuration obtained after node $10$ makes this choice executing a \emph{SingleNode} move: since $\mathit{Lowest} \{u\in \Voisin(10)~|~p_u = 10\}=2$, node $10$ points to node $2$.  Now, node $24$ is eligible to execute a \emph{MatchSecond} move and, since $BestRematch(9)$ has changed, node $9$ is eligible to execute an \emph{Update} move.

Let us assume node $24$ is activated (see Figure~\ref{fig:execution}.(d) for configuration after this move). It then points to node $15$ thus, node $15$ can accept the proposition executing a \emph{SingleNode} move.  So, it does it and it sets $p_{15} =24$.
 Figure~\ref{fig:execution}.(e) shows after this moves.

Since $p_{10}\neq 9$ and $(\alpha_9, \beta_9)\neq BestRematch(9)$, node $9$ can execute an \emph{Update} move. Figure~\ref{fig:execution}.(f) shows the configuration obtained after this move: $(\alpha_9, \beta_9)=(null,null)$ and $(p_9,s_9)=(null,false)$.  This will cause $AskFirst(8,9)=AskSecond(8,9)=null$.  Then node $8$ executes a \emph{ResetMatch} move (see configuration after this move Figure~\ref{fig:execution}.(g)). This will cause node $7$ to execute a \emph{SingleNode} move and sets $p_7 =null$ as seen in Figure~\ref{fig:execution}.(h). The system then has reached a stable configuration. Thus, the size of the matching is increasing by one and only one augmenting path has been fully exploited.

 \begin{figure}
  \begin{tabular}{p{8cm}p{8cm}}
     \scalebox{0.65}{
\begin{tikzpicture}[xscale = -0.8,yscale=0.8]
	\tikzstyle{legend}=[rectangle, thin,minimum width=2.5cm, minimum height=0.8cm,black]
	\tikzstyle{vertex}=[draw,circle,minimum size=0.3cm,fill=black,inner sep=0pt]
	\tikzstyle{vertex retour}=[draw,rectangle,minimum size=0.25cm,fill=black,inner sep=0pt]
	\tikzstyle{matched edge} = [draw,line width=3pt,-,black]
	\tikzstyle{edge} = [draw,thick,-,black!40]
	\tikzstyle{pointer} = [draw, thick,->,>=stealth,bend left=40,black,shorten >=1pt]


\foreach \x in {3,4,...,9}
{\node[vertex] (v-\x) at (\x,0.5) {};}


\foreach \from/\to in {3/4,4/5,5/6,6/7,7/8,8/9}
{       \draw[edge] (v-\from)--(v-\to);}


\foreach \from/\to in {4/5,7/8}
 {       \draw[matched edge] (v-\from)to (v-\to);}

\foreach \x in {7,8,...,10}
\node[legend] (b1) at (\x-4,0) {\x};
 \node[legend] (b1) at (11-2,0) {15};
\node[legend] (b1) at ( 8,0) {24};
\node[legend] (b1) at ( 7,0) {2};

\draw[pointer] (v-3) to (v-4);
\draw[pointer] (v-4) to (v-3);
\draw[pointer] (v-5) to (v-6);
 
\end{tikzpicture}
} (a) Initial configuration.
 &
  \scalebox{0.65}{
\begin{tikzpicture}[xscale = -0.8,yscale=0.8]
	\tikzstyle{legend}=[rectangle, thin,minimum width=2.5cm, minimum height=0.8cm,black]
	\tikzstyle{vertex}=[draw,circle,minimum size=0.3cm,fill=black,inner sep=0pt]
	\tikzstyle{vertex retour}=[draw,rectangle,minimum size=0.25cm,fill=black,inner sep=0pt]
	\tikzstyle{matched edge} = [draw,line width=3pt,-,black]
	\tikzstyle{edge} = [draw,thick,-,black!40]
	\tikzstyle{pointer} = [draw, thick,->,>=stealth,bend left=40,black,shorten >=1pt]


\foreach \x in {3,4,...,9}
{\node[vertex] (v-\x) at (\x,0.5) {};}


\foreach \from/\to in {3/4,4/5,5/6,6/7,7/8,8/9}
{       \draw[edge] (v-\from)--(v-\to);}


\foreach \from/\to in {4/5,7/8}
 {       \draw[matched edge] (v-\from)to (v-\to);}

\foreach \x in {7,8,...,10}
\node[legend] (b1) at (\x-4,0) {\x};
  \node[legend] (b1) at (11-2,0) {15};
\node[legend] (b1) at ( 8,0) {24};
\node[legend] (b1) at ( 7,0) {2};

\draw[pointer] (v-3) to (v-4);
\draw[pointer] (v-4) to (v-3);
\draw[pointer] (v-5) to (v-6);
 
\draw[pointer] (v-8) to (v-9);
\draw[pointer] (v-9) to (v-8);
\draw[pointer] (v-7) to (v-6);
\draw[pointer] (v-6) to (v-7);

\end{tikzpicture}
}
  
       (e)  Node $15$    executes a \emph{SingleNode} move. \\
  \scalebox{0.65}{
\begin{tikzpicture}[xscale = -0.8,yscale=0.8]
	\tikzstyle{legend}=[rectangle, thin,minimum width=2.5cm, minimum height=0.8cm,black]
	\tikzstyle{vertex}=[draw,circle,minimum size=0.3cm,fill=black,inner sep=0pt]
	\tikzstyle{vertex retour}=[draw,rectangle,minimum size=0.25cm,fill=black,inner sep=0pt]
	\tikzstyle{matched edge} = [draw,line width=3pt,-,black]
	\tikzstyle{edge} = [draw,thick,-,black!40]
	\tikzstyle{pointer} = [draw,  thick,->,>=stealth,bend left=40,black,shorten >=1pt]


\foreach \x in {3,4,...,9}
{\node[vertex] (v-\x) at (\x,0.5) {};}


\foreach \from/\to in {3/4,4/5,5/6,6/7,7/8,8/9}
{       \draw[edge] (v-\from)--(v-\to);}


\foreach \from/\to in {4/5,7/8}
 {       \draw[matched edge] (v-\from)to (v-\to);}

\foreach \x in {7,8,...,10}
\node[legend] (b1) at (\x-4,0) {\x};
 \node[legend] (b1) at (11-2,0) {15};
\node[legend] (b1) at ( 8,0) {24};
\node[legend] (b1) at ( 7,0) {2};

\draw[pointer] (v-3) to (v-4);
\draw[pointer] (v-4) to (v-3);
\draw[pointer] (v-5) to (v-6);
 
\draw[pointer] (v-7) to (v-6);

\end{tikzpicture}
}
     (b) Node $2$   executes a \emph{MatchFirst}
 move. &      \scalebox{0.65}{
\begin{tikzpicture}[xscale = -0.8,yscale=0.8]
	\tikzstyle{legend}=[rectangle, thin,minimum width=2.5cm, minimum height=0.8cm,black]
	\tikzstyle{vertex}=[draw,circle,minimum size=0.3cm,fill=black,inner sep=0pt]
	\tikzstyle{vertex retour}=[draw,rectangle,minimum size=0.25cm,fill=black,inner sep=0pt]
	\tikzstyle{matched edge} = [draw,line width=3pt,-,black]
	\tikzstyle{edge} = [draw,thick,-,black!40]
	\tikzstyle{pointer} = [draw, thick,->,>=stealth,bend left=40,black,shorten >=1pt]


\foreach \x in {3,4,...,9}
{\node[vertex] (v-\x) at (\x,0.5) {};}


\foreach \from/\to in {3/4,4/5,5/6,6/7,7/8,8/9}
{       \draw[edge] (v-\from)--(v-\to);}


\foreach \from/\to in {4/5,7/8}
 {       \draw[matched edge] (v-\from)to (v-\to);}

\foreach \x in {7,8,...,10}
\node[legend] (b1) at (\x-4,0) {\x};
  \node[legend] (b1) at (11-2,0) {15};
\node[legend] (b1) at ( 8,0) {24};
\node[legend] (b1) at ( 7,0) {2};

\draw[pointer] (v-3) to (v-4);
\draw[pointer] (v-4) to (v-3);
 
\draw[pointer] (v-8) to (v-9);
\draw[pointer] (v-9) to (v-8);
\draw[pointer] (v-7) to (v-6);
\draw[pointer] (v-6) to (v-7);

\end{tikzpicture}
}
     (f) Node $9$ executes a \emph{Update} move. \\
\\
     \scalebox{0.65}{
\begin{tikzpicture}[xscale = -0.8,yscale=0.8]
	\tikzstyle{legend}=[rectangle, thin,minimum width=2.5cm, minimum height=0.8cm,black]
	\tikzstyle{vertex}=[draw,circle,minimum size=0.3cm,fill=black,inner sep=0pt]
	\tikzstyle{vertex retour}=[draw,rectangle,minimum size=0.25cm,fill=black,inner sep=0pt]
	\tikzstyle{matched edge} = [draw,line width=3pt,-,black]
	\tikzstyle{edge} = [draw,thick,-,black!40]
	\tikzstyle{pointer} = [draw,  thick,->,>=stealth,bend left=40,black,shorten >=1pt]


\foreach \x in {3,4,...,9}
{\node[vertex] (v-\x) at (\x,0.5) {};}


\foreach \from/\to in {3/4,4/5,5/6,6/7,7/8,8/9}
{       \draw[edge] (v-\from)--(v-\to);}


\foreach \from/\to in {4/5,7/8}
 {       \draw[matched edge] (v-\from)to (v-\to);}

\foreach \x in {7,8,...,10}
\node[legend] (b1) at (\x-4,0) {\x};
  \node[legend] (b1) at (11-2,0) {15};
\node[legend] (b1) at ( 8,0) {24};
\node[legend] (b1) at ( 7,0) {2};

\draw[pointer] (v-3) to (v-4);
\draw[pointer] (v-4) to (v-3);
\draw[pointer] (v-5) to (v-6);
 
\draw[pointer] (v-7) to (v-6);
\draw[pointer] (v-6) to (v-7);
\end{tikzpicture}
}
   (c) Node $10$ executes A
 \emph{SingleNode} move.  
&      \scalebox{0.65}{
\begin{tikzpicture}[xscale = -0.8,yscale=0.8]
	\tikzstyle{legend}=[rectangle, thin,minimum width=2.5cm, minimum height=0.8cm,black]
	\tikzstyle{vertex}=[draw,circle,minimum size=0.3cm,fill=black,inner sep=0pt]
	\tikzstyle{vertex retour}=[draw,rectangle,minimum size=0.25cm,fill=black,inner sep=0pt]
	\tikzstyle{matched edge} = [draw,line width=3pt,-,black]
	\tikzstyle{edge} = [draw,thick,-,black!40]
	\tikzstyle{pointer} = [draw,  thick,->,>=stealth,bend left=40,black,shorten >=1pt]


\foreach \x in {3,4,...,9}
{\node[vertex] (v-\x) at (\x,0.5) {};}


\foreach \from/\to in {3/4,4/5,5/6,6/7,7/8,8/9}
{       \draw[edge] (v-\from)--(v-\to);}


\foreach \from/\to in {4/5,7/8}
 {       \draw[matched edge] (v-\from)to (v-\to);}

\foreach \x in {7,8,...,10}
\node[legend] (b1) at (\x-4,0) {\x};
  \node[legend] (b1) at (11-2,0) {15};
\node[legend] (b1) at ( 8,0) {24};
\node[legend] (b1) at ( 7,0) {2};

\draw[pointer] (v-3) to (v-4);
 
\draw[pointer] (v-8) to (v-9);
\draw[pointer] (v-9) to (v-8);
\draw[pointer] (v-7) to (v-6);
\draw[pointer] (v-6) to (v-7);
\end{tikzpicture}
\label{exemple-chaine-final}
}
 
(g) Node $8$  executes a
    \emph{ResetMatching} move \\
      \scalebox{0.65}{
\begin{tikzpicture}[xscale = -0.8,yscale=0.8]
	\tikzstyle{legend}=[rectangle, thin,minimum width=2.5cm, minimum height=0.8cm,black]
	\tikzstyle{vertex}=[draw,circle,minimum size=0.3cm,fill=black,inner sep=0pt]
	\tikzstyle{vertex retour}=[draw,rectangle,minimum size=0.25cm,fill=black,inner sep=0pt]
	\tikzstyle{matched edge} = [draw,line width=3pt,-,black]
	\tikzstyle{edge} = [draw,thick,-,black!40]
	\tikzstyle{pointer} = [draw, thick,->,>=stealth,bend left=40,black,shorten >=1pt]


\foreach \x in {3,4,...,9}
{\node[vertex] (v-\x) at (\x,0.5) {};}


\foreach \from/\to in {3/4,4/5,5/6,6/7,7/8,8/9}
{       \draw[edge] (v-\from)--(v-\to);}


\foreach \from/\to in {4/5,7/8}
 {       \draw[matched edge] (v-\from)to (v-\to);}

\foreach \x in {7,8,...,10}
\node[legend] (b1) at (\x-4,0) {\x};
  \node[legend] (b1) at (11-2,0) {15};
\node[legend] (b1) at ( 8,0) {24};
\node[legend] (b1) at ( 7,0) {2};

\draw[pointer] (v-3) to (v-4);
\draw[pointer] (v-4) to (v-3);
\draw[pointer] (v-5) to (v-6);
 
\draw[pointer] (v-8) to (v-9);
\draw[pointer] (v-7) to (v-6);
\draw[pointer] (v-6) to (v-7);

\end{tikzpicture}
}
  
       (d)  Node $24$  executes a
    \emph{MatchSecond} move.   &   \scalebox{0.65}{
\begin{tikzpicture}[xscale = -0.8,yscale=0.8]
	\tikzstyle{legend}=[rectangle, thin,minimum width=2.5cm, minimum height=0.8cm,black]
	\tikzstyle{vertex}=[draw,circle,minimum size=0.3cm,fill=black,inner sep=0pt]
	\tikzstyle{vertex retour}=[draw,rectangle,minimum size=0.25cm,fill=black,inner sep=0pt]
	\tikzstyle{matched edge} = [draw,line width=3pt,-,black]
	\tikzstyle{edge} = [draw,thick,-,black!40]
	\tikzstyle{pointer} = [draw,  thick,->,>=stealth,bend left=40,black,shorten >=1pt]


\foreach \x in {3,4,...,9}
{\node[vertex] (v-\x) at (\x,0.5) {};}


\foreach \from/\to in {3/4,4/5,5/6,6/7,7/8,8/9}
{       \draw[edge] (v-\from)--(v-\to);}


\foreach \from/\to in {4/5,7/8}
 {       \draw[matched edge] (v-\from)to (v-\to);}

\foreach \x in {7,8,...,10}
\node[legend] (b1) at (\x-4,0) {\x};
  \node[legend] (b1) at (11-2,0) {15};
\node[legend] (b1) at ( 8,0) {24};
\node[legend] (b1) at ( 7,0) {2};

 
\draw[pointer] (v-8) to (v-9);
\draw[pointer] (v-9) to (v-8);
\draw[pointer] (v-7) to (v-6);
\draw[pointer] (v-6) to (v-7);
\end{tikzpicture}
\label{exemple-chaine-final}
}
   (h)  Node $7$ executes a
 \emph{SingleNode} move. \\
   \label{fig:exemple-chaine} 
 \end{tabular}
 \caption{An execution of Algorithm $\mathcal{M}^+$ }
\label{fig:execution}
 \end{figure}
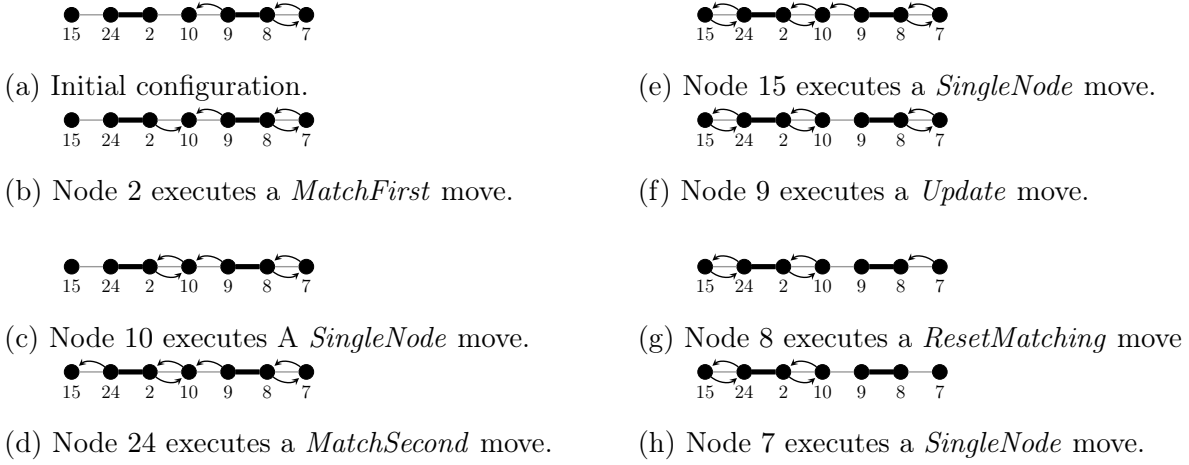

\subsection{Description of Algorithm $\mathcal{M}^+$}

\begin{definition}[Edges in state \emph{On} or in state  \emph{Off}] \label{def:edge}
Let $e=(u,v)$ be an edge in the maximal matching $M$.
Let $x$ (resp. $y$) be the single node adjacent to $u$ (resp. $v$). 
Edge $e$ is said to be   \emph{in state  Off}  if $p_{u}=\bot$, $p_{v}=\bot$, $p_x=\bot$ and $p_y=\bot$.   
Moreover edge $e$ is said to be \emph{in state  On}  if  $p_{u}=y,p_{x}=v$, $p_{v}=x$ and $p_y = null$. Edge $e$ is said to be \emph{in state Almost On}  if  $p_{u}=y,p_{x}=v$, $p_{v}=x$ and $p_y \notin\{null,u\}$. 
\end{definition}

\begin{figure}[!hbt]
 \centering
 \scalebox{0.5}{
\begin{tikzpicture}[scale=0.6]
	\tikzstyle{legend}=[rectangle, thin,minimum width=2.5cm, minimum height=1.2cm,black]
	\tikzstyle{node}=[draw,circle,minimum size=0.6cm,thick,inner sep=0pt]
	\tikzstyle{matched edge} = [draw,line width=3pt,-,black]
	\tikzstyle{edge} = [draw,-,black]
	\tikzstyle{pointer} = [draw, thick,->,>=stealth,black,shorten >=1pt]
	    \node[node] (v0) at  (0,0)  {$y$};
	    \node[node] (v1) at  (2.5,0)  {$u$};
	    \node[node] (v2) at  (5,0) {$v$};
	    \node[node] (v3) at  (7.5,0) {$x$};
	    \draw[matched edge] (v1) -- (v2);
	    \draw[edge] (v0) -- (v1);
	    \draw[edge] (v2) -- (v3);

	    \node[node] (v5) at  (0+12.5,0)  {$y$};
	    \node[node] (v6) at  (2.5+12.5,0)  {$u$};
	    \node[node] (v7) at  (5+12.5,0) {$v$};
	    \node[node] (v8) at  (7.5+12.5,0) {$x$};
	    \draw[matched edge] (v6) -- (v7);
	    \draw[edge] (v5) -- (v6);
	    \draw[edge] (v7) -- (v8);
	    \draw[pointer] (v8) to  [out=-125,in=-45] (v7);
	    \draw[pointer] (v6) to [out=-125,in=-45] (v5);
	    \draw[pointer] (v7) to [out=55,in=135] (v8);

\node[node] (v15) at  (0+27,0)  {$y$};
	    \node[node] (v16) at  (2.5+27,0)  {$u$};
	    \node[node] (v17) at  (5+27,0) {$v$};
	    \node[node] (v18) at  (7.5+27,0) {$x$};
	    \node[node,  dashed] (v19) at  (-2.5+27,0) {}; 

	    \draw[matched edge] (v16) -- (v17);
	    \draw[edge] (v15) -- (v16);
	    \draw[edge] (v17) -- (v18);
	    \draw[edge,densely dashed] (v15) -- (v19);

	    \draw[pointer] (v18) to  [out=-125,in=-45] (v17);
	    \draw[pointer] (v16) to [out=-125,in=-45] (v15);
	    \draw[pointer] (v17) to [out=55,in=135] (v18);
	    \draw[pointer] (v15) to [out=-125,in=-45] (v19);

\node[legend] (l1) at (3.5,-1.5) {State Off};
\node[legend] (l2) at (3.5+12.5,-1.5) {State On};
\node[legend] (l2) at (3.5+27,-1.5) {State almost On};

\end{tikzpicture}
}
 \caption{Edges  in state Off and On :
the arrows drawn represent the local variables $p_{\cdot}$ of nodes.}
\label{fig:onoff}
 \end{figure}
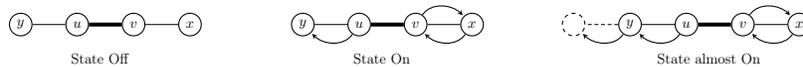

 An example of Definition~\ref{def:edge} can be seen in Figure \ref{fig:onoff}. Moreover, in Figure \ref{fig:execution}.(a), edge $(24,2)$ is in state \emph{Off} while  edge $(9,8)$ is in state \emph{On}.

The states of edges represent the detection process step of the $3$-augmenting path. Now, we will exhibit an execution to switch edge $(u,v)$ from state \emph{Off} to state \emph{On}.

\begin{lemma}\label{lem:offTOon}
Let $e=(u,v)$ be an edge in the maximal matching $M$ and in state Off.  Let $y$ (resp. $x$) be the single node adjacent to $u$  (resp. $v$) with $Ident(x) <Ident(y)$. If $y= Lowest\{BestRematch(u)\}$, $x=Lowest\{BestRematch(v)\}$ and $v \leq Lowest \{w\in
  \Voisin(x)  |p_x =w\}$, then there exists a finite execution to switch edge $(u,v)$ from state \emph{Off} to state \emph{On}.  Moreover the only nodes executing a move in this execution are $ \{x,u,v\}$. 
\end{lemma}

\begin{proof} We describe a finite execution to switch edge $(u,v)$ from state Off to state On.  Nodes $u$ and $v$ belong to a  $3$-augmenting path since $p_x=p_y=null$. If $\alpha_u\neq y$, then  node $u$ executes a \emph{Update} move :
$(\alpha_u,\beta_u)=(y,null)$ because $p_y=null$. If $\alpha_v\neq x$, then  node $v$ executes a \emph{Update} move :
$(\alpha_v,\beta_v)=(x,null)$ because $p_x=null$ and  $v = Lowest \{w\in \Voisin(x)  |p_x =w\}$.

Thus,  the variables $\alpha_u$ and $\alpha_v$ are well defined : $\alpha_u= y$ and $\alpha_v= x$. $Ident(x) <Ident(y)$ implies $AskFirst(v,u) = x$ and $AskFirst(u,v) = null$ because $2 \leq Unique( \{\alpha_{u}, \beta_{u}, \alpha_{v}, \beta_{v}\} ) \leq 4$. Thus node $v$ executes a \emph{MatchFirst} move: $p_v=x$. Since $v=Lowest\{w\in \Voisin(x)  |p_w =x\}$ by the hypothesis of this lemma, node $x$ chooses node $v$ to match with by executing a \emph{SingleNode} move.  Finally, node $u$ is eligible to execute a \emph{MatchSecond} move and it then points to node $y$ (because $y= Lowest\{BestRematch(u)\}$).
\end{proof}

Note that  Figures \ref{fig:execution}.(a)- \ref{fig:execution}.(d) represent an execution to switch edge $(24,2)$ from state emph{Off} to state \emph{On}: nodes $15$ and $10$ are respectively nodes   $y$ and $x$ for the execution of Lemma~\ref{lem:onTOoff}. Now,   Now, we will exhibit an execution to switch edge $(u,v)$ from state \emph{Almost On} to state \emph{Off}.

\begin{lemma}\label{lem:onTOoff}
Let $e=(u,v)$ be an edge in the maximal matching $M$ and in state \emph{Almost On}.  Let $y$ (resp. $x$) be the single node adjacent to $u$ (resp. $v$) with $Ident(x) <Ident(y)$.  There exists a finite execution to switch edge $(u,v)$ from state \emph{Almost On} to state \emph{Off}.  Moreover the only nodes executing a move in this execution are $ \{x,y,u,v\}$.
\end{lemma}

\begin{proof} 
A finite execution to switch edge $(u,v)$ from state \emph{Almost On} to state \emph{Off} is described. Since edge $(u,v)$ is in state \emph{Almost On},  $p_{y}\not\in \{u,null\}$ and so $(\alpha_u, \beta_u)\neq BestRematch(u)$.  Node $u$ executes a \emph{Update} move. After this move, $(p_u,s_u)=(null,false)$.  The fact that $\alpha_u=null$ will cause $AskFirst(v,u)=AskSecond(v,u)=null$. Then node $v$ executes a \emph{ResetMatch} move: $p_v=null$. Then node $x$ is activated by executing a \emph{SingleNode} move and it sets $p_x=null$. Finally, node $v$ can execute  a \emph{Update} move, and thus $(\alpha_v, \beta_v)=(null,null)$. 
\end{proof}

Note that in Figure~\ref{fig:execution}.(e),  edge $(9,8)$ is in state \emph{Almost On}. Figures~\ref{fig:execution}.(f)-\ref{fig:execution}.(h) represent the execution of Lemma~\ref{lem:onTOoff} in order that  edge $(9,8)$ will be in state
\emph{Off}.

\subsection{Complexity of Algorithm $\mathcal{M}^+$}

We  describe an execution corresponding to count from 0 to $2^{N}-1$, where $N$ is an arbitrary integer.
This execution occurs in a graph denoted by $G_N$ with $\Theta(N^2)$ nodes.  $G_N$ is composed in $N$
sub-graphs, each of them representing a bit.  The whole graph then represents an integer, coding from theses $N$ bits. 
$G_N$ has 2 kind of nodes: the nodes represented by circles ($\rond\,$-nodes)  and those represented by squares ($\carre\,$-nodes). The $\rond\,$-nodes are used to store bits value and hence an integer. The $\carre\,$-nodes are used to implement the ``$+1$'' operation as we count from $0$ to $2^N-1$.

\paragraph{Example:}
As an illustration, graph $G_4$ is shown in Figure~\ref{figure-010}. In this example, the bold edges are those that belong to the maximal matching $M$ computed by algorithm $\mathcal{M}$ and arrows represent the local variable $p$ of the $2/3$-approximation algorithm. A node having no outgoing arrow has its $p$ variable equals to \emph{null}. 

 \begin{figure}[!hbt]
 \centering
 \scalebox{0.6}{
\begin{tikzpicture}[xscale = -0.8,yscale=0.8]
	\tikzstyle{legend}=[rectangle, thin,minimum width=2.5cm, minimum height=0.8cm,black]
	\tikzstyle{vertex}=[draw,circle,minimum size=0.3cm,fill=black,inner sep=0pt]
	\tikzstyle{vertex retour}=[draw,rectangle,minimum size=0.25cm,fill=black,inner sep=0pt]
	\tikzstyle{matched edge} = [draw,line width=3pt,-,black]
	\tikzstyle{edge} = [draw,thick,-,black!40]
	\tikzstyle{pointer} = [draw, thick,->,>=stealth,bend left=40,black,shorten >=1pt]

\foreach \x in {3,4,5,6,9,10,11,12,15,16,17,18,21,22,23,24}
{\node[vertex] (v-\x) at (\x,0.5) {};}

\foreach \from/\to in {3/4,5/6,9/10,11/12,15/16,17/18,21/22,23/24}
{       \draw[edge] (v-\from)--(v-\to);}
\foreach \from/\to in {4/5,10/11,16/17,22/23}
{       \draw[matched edge] (v-\from)to (v-\to);}

\foreach \x in {7,8,...,10}
\node[legend] (b1) at (\x-4,0) {\x};
\foreach \x in {11,12,...,14}
\node[legend] (b1) at (\x-2,0) {\x};
\foreach \x in {15,16,...,18}
\node[legend] (b1) at (\x,0) {\x};

\foreach \x in {19,20,...,22}
\node[legend] (b1) at (\x+2,0) {\x};


\node[vertex retour] (r-32) at (20,-1) {};  \node[legend] (b1) at (20,-1.5) {28};
\node[vertex retour] (r-31) at (20,-2) {}; \node[legend] (b1) at (20,-2.5) {27};
\node[vertex retour] (r-30) at (20,-3) {}; \node[legend] (b1) at (20,-3.5) {26};
\node[vertex retour] (r-21) at (14,-1) {}; \node[legend] (b1) at (14,-1.5) {25};
\node[vertex retour] (r-20) at (14,-2) {}; \node[legend] (b1) at (14,-2.5) {24};
\node[vertex retour] (r-10) at ( 8,-1) {}; \node[legend] (b1) at ( 8,-1.5) {23};
\node[vertex retour] (l-32) at (19,-1) {}; \node[legend] (b1) at (19,-1.5) {6};
\node[vertex retour] (l-31) at (19,-2) {}; \node[legend] (b1) at (19,-2.5) {5};
\node[vertex retour] (l-30) at (19,-3) {};  \node[legend] (b1) at (19,-3.5) {4};
\node[vertex retour] (l-20) at (13,-2) {}; \node[legend] (b1) at (13,-2.5) {2};
\node[vertex retour] (l-10) at ( 7,-1) {};   \node[legend] (b1) at ( 7,-1.5) {1};

\node[vertex retour] (l-21) at (13,-1) {};  \node[legend] (b1) at ( 13,-1.5) {3};

\foreach \x in {30,31,32,21,20,10}
{\draw[matched edge] (r-\x)--(l-\x);}
\draw[edge] (v-9)--(r-10);
\foreach \x in {21,20}
{\draw[edge] (v-15)--(r-\x);}
\foreach \x in {31,30,32}
{\draw[edge] (v-21)--(r-\x);}
\foreach \x in {10,20,30}
{\draw[edge] (v-6)--(l-\x);}
\foreach \x in {21,31}
{\draw[edge] (v-12)--(l-\x);}
\draw[edge] (v-18)--(l-32);
\foreach \x in {6.5,12.5,18.5}
{\draw [loosely dotted,thick] (\x,-4) -- (\x,2) ;}

\draw[pointer] (v-9) to (v-10);
\draw[pointer] (v-10) to (v-9);
\draw[pointer] (v-11) to (v-12);
\node[legend] (b1) at (4,1.5) {Bit 0 = 0};
\node[legend] (b2) at (10,1.5) {Bit 1 =1};
\node[legend] (b3) at (16,1.5) {Bit 2 =0};
\node[legend] (b4) at (22,1.5) {Bit 3 =0};

\end{tikzpicture}
}
 \caption{Graph   $G_4$ encoding  0010}
\label{figure-010}
 \end{figure}
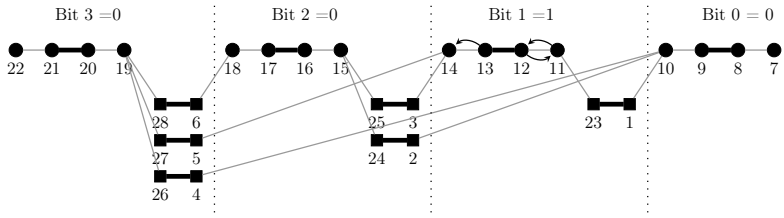 

As we said, the $\rond$-nodes are used to encode the $N$ bits. Each bit $i$ is encoded with the local state of the $4$ following nodes: $\B{i}{1},\B{i}{2},\B{i}{3},\B{i}{4}$. These nodes are then named $\B{i}{k}$, for ``\emph{the $k^{th}$ node of the bit i}''. For instance, node $10$ is the fourth node of the bit $1$, thus $10$ is called $\B{1}{4}$. In the following, we will refer to these four nodes as the $i^{th}$ \emph{bit-block}.

A binary value can be associated to each bit-block according to the $p$-values of each nodes in the bit-block. We will formally define  this association later, but we can already say that in this example, according to the $p$-value of all the nodes in the 4 bit-blocks, $G_4$ encodes the binary integer  $0010$.

\paragraph{$\mathbf{G_N}$ definition:}
In the following, we formally describe the graph $G_N=(V_N,E_N)$.
\begin{enumerate}
\item  $V_N= V_N^{\rond} \cup V_N^{\carre} $ where
\begin{eqnarray*}
V_N^{\rond}&=&\bigcup_{0\leq i<N}\{\B{i}{k}|k=1,2,3,4\}\\
V_N^{\carre}&=&\bigcup_{0\leq j< i<N}\{\Rr{i}{j},\Rl{i}{j}\}\\
\end{eqnarray*}
\item $E_N=E_N^{\rond} \cup E_N^{\carre} $ where
\begin{eqnarray*}
E_N^{\rond}&=&\bigcup_{0\leq i <N}\{(\B{i}{k},\B{i}{k+1})|k=1,2,3\} \\
E_N^{\carre}&=&\bigcup_{0\leq j < i <N}\{(\B{i}{1},\Rr{i}{j})~,~(\Rr{i}{j},\Rl{i}{j})~,~(\Rl{i}{j},\B{j}{4})\}
\end{eqnarray*}
\end{enumerate}

\noindent Figure \ref{fig:example} gives a partial view of the graph $G_N$ corresponding to the $i$th bit-block.
\begin{figure}[htbp]
   \centering
   \includegraphics[width=\textwidth]{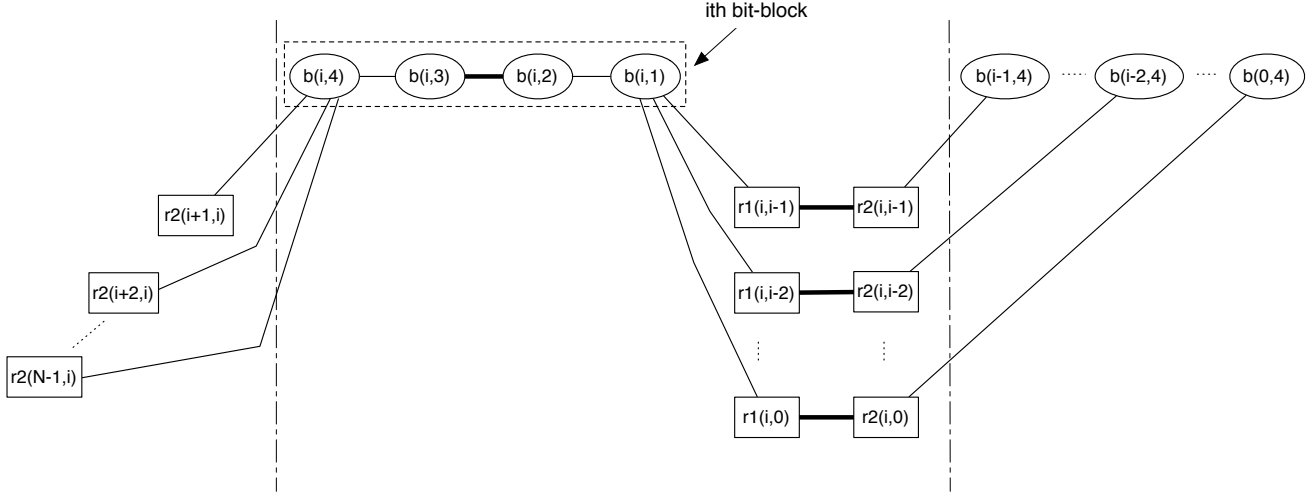} 
   \caption{A partial view of  graph $G_N$}
   \label{fig:example}
\end{figure}

Our execution is based on the maximal matching $M$ computed by the algorithm $\mathcal{M}$:
$${M}=\{(\B{i}{2},\B{i}{3})|0\leq i<N\} \cup \{ (\Rr{i}{j},\Rl{i}{j})| 0\leq j<i<N \}$$
This maximal matching $M$ is encoded with the $m$-variable. Then we have:
$$m_{\B{i}{2}}=\B{i}{3} , m_{\B{i}{3}}=\B{i}{2} ,m_{\Rr{i}{j}}=\Rl{i}{j}\text{ and }m_{\Rl{i}{j}}=\Rr{i}{j}$$
This matching is an $\frac{1}{2}$-approximation of the maximum matching and the algorithm $\mathcal{M}^+$ updates this approximation  building a $\frac{2}{3}$-approximation of the maximum matching based on $M$.  This $\frac{2}{3}$-approximation is encoded with the $p$-variable in $\mathcal{M}^+$. We also use the variable $p$  to encode a bit associated to a bit-block. The two following definitions give this association:

\begin{definition}[Bit-block encoding]
In graph $G_N$, let $\{\B{i}{1},\B{i}{2},\B{i}{3},\B{i}{4}\}$ be the
$i^{th}$ bit-block, for some $0\leq i<N$. This bit-block encodes the
value 1 (\emph{resp.} 0) if the edge $(\B{i}{2},\B{i}{3})$ is in state
\emph{On} (\emph{resp.} Off).
\end{definition}
Note that the value is not always defined.  We can associate an integer $\omega$ to such a configuration of the graph. 

\begin{definition}[$\omega$-configuration]
Let $\omega$ be represented the integer  such that $\omega <2^N$, a
configuration is said to be an \emph{$\omega$-configuration}   if for
any integer $i \leq N$, the $i^{th}$ bit of $\omega$ is the value encoded by the $i^{th}$ block of nodes.
\end{definition} 

Figure \ref{figure-010} shows a $3$-configuration.

\paragraph{Identifiers:}
In order to exhibit our execution counting from $0$ to $2^N-1$, we need to be able to switch edges between \emph{on} and \emph{off}. This can be done executing the guarded rules of $\mathcal{M}^+$. Since this algorithm uses identifiers, we need some properties on identifiers of nodes in $G_N$. The $ident$ function gives the identifier associated to a node in $V_N$. We assume each node has a unique identifier.  These identifiers must satisfy the three following properties:

\begin{property}[Identifiers order] \label{identOrder}
Let $\B{i}{k}, \B{i'}{k'}, \B{i}{2}$ and $\B{i}{3}$ be nodes in  $V_N^{\rond}$, and $\Rr{i}{j}$ and $\Rl{i}{j}$ be nodes in $V_N^{\carre}$. We have:
\begin{enumerate}
\item $ident(\B{i}{k})>ident(\B{i'}{k'})\textrm{  if } (i>i') \lor (i=i'\land k>k')$
\item $ident(\B{i}{2})<ident(\Rr{i}{j})$
\item $ident(\B{i}{3})>ident(\Rl{j}{i})$
\end{enumerate}
\end{property}

\noindent Note that in  graph $G_N$, it exists an \emph{ident} function that satisfies Property \ref{identOrder}. For instance, the property holds for the following naming:\\
Let $c = |V_N^{\rond}|$ and $s = \frac{|V_N^{\carre}|}{2}$. There are $c$ nodes of kind $b$, $s$ nodes of kind $r_1$ and $s$ nodes of kind $r_2$ as well. 
\begin{itemize}
\item Nodes of kind $r_2$ are named from $1$ to $s$
\item Nodes of kind $b$ are named from $s+1$ to $s+c$ such that:\\
$$\forall i, 0\leq i <N, \forall k\in \{1,2,3,4\}: ident(b(i,k)) = s+i+k$$
\item Nodes of kind $r_1$ are named from $s+c+1$ to $s+c+s$
\end{itemize}

Figure \ref{figure-010} shows graph $G_4$ with such a naming.

\paragraph{Counting from $\mathbf{0}$ to $\mathbf{2^N-1}$:} We will build an execution containing all $\omega$-configurations with $1\leq \omega <2^N $  To to this, we will build an execution from $\omega$-configuration to $(\omega+1)$-configuration using ``$+1$'' operation. This allows for the counting from $\mathbf{0}$ to $\mathbf{2^N-1}$. As we said before, the nodes in $V_N^{\carre}$ are used to implement the ``$+1$'' operation. To do that, we need to be able to switch bit from 0 to 1 and from 1 to
0, in a clever way. To switch from 0 to 1 is easier than to switch from 1 to 0. The nodes in $V_N^{\carre}$ are used to implement the switch from 1 to 0. The main scheme is the following: let us consider a binary integer $x$. The '+1' operation consists in finding the
rightmost 0 in $x$. Then all 1 at the right of this 0 have to switch to 0 and this 0 has to switch to 1 (if $x=x' 011\ldots1$ then
$x+1=x'100\ldots0$).  Let us assume that 0 is the $i^{th}$ bit of $x$. The $i^{th}$ bit-block has to switch from 0 to 1 during the '+1'
operation. Afterwards,  each $j^{th}$ bit-block, with $0\leq j< i$, has to switch from 1 to 0. To perform this switch, we use vertices in $V_N^{\carre}$.

We will now describe a piece of the execution, starting on the configuration represented on Figure \ref{figure-010}. The graph drawn in this figure encodes integer $(0010)$. We illustrate the use of the \emph{$p$-values} by an arrow and the absence of the arrow means that the \emph{$p$-value} of the node equals to $null$. First, we will focus on vertices in the $0$th bit-block. Edge
$(\B{0}{2},\B{0}{3})$ belongs to the underlying maximal matching represented by bold edges and is in state \emph{Off}.  Lemma
\ref{lem:offTOon}, describes an execution from the $0010$-configuration represented on Figure~\ref{figure-010} to the $0011$-configuration represented on Figure~\ref{figure-011}.  Moreover,  Figures \ref{figure-011}, \ref{figure-011-1}, \ref{figure-011-2}, \ref{figure-011-3} and \ref{figure-100} illustrate the transformation from $0011$-configuration to $0111$-configuration in graph $G_4$.

\begin{theorem}\label{th:execution:+1}
Let  $\omega$ be an integer such that $\omega < 2^N-1$. There exists a finite execution to transform an $\omega$-configuration 
into an $(\omega+1)$-configuration. 
\end{theorem}

\begin{proof}
Let $i$ be the integer such that the $i-1$ first bits of $\omega$ equal to $1$ and the value of its $i^{th}$ bit to $0$. This implies that the $ith$ bit of $\omega+1$ bits is the first bit equal to $1$. 

We distinguish two cases : $i=0$ and $i>0$.

In the case where $i=0$, edge $(\B{0}{2},\B{0}{3})$ is in state \emph{Off} by definition.  Since the $0th$ bit of integer $\omega+1$ is equal to $1$, $(\B{0}{2},\B{0}{3})$ is in state \emph{On} in $(\omega+1)$-configuration. By  Property \ref{identOrder},  we have $Ident(\B{0}{1}) <Ident(\B{0}{4})$ and by definition of edge in state \emph{Off}, $p_{\B{0}{1}}=p_{\B{0}{4}}=null $.   Note that   $ \B{0}{2}  \leq Lowest \{w\in \Voisin(\B{0}{1})   |p_{\B{0}{1}} =w\}$. Since nodes $\B{0}{3} $ and $\B{0}{2}$ only have one Single node as neighbor,  the hypotheses of  Lemma \ref{lem:offTOon}  are satisfied 
From Lemma \ref{lem:offTOon}, there exists an execution to switch edge $(\B{0}{2},\B{0}{3})$ from state \emph{Off} to state \emph{On}.  At the end, the least significant bit of the integer correspond to this current configuration is set to $1$. So we obtain a $(\omega+1)$-configuration. 

 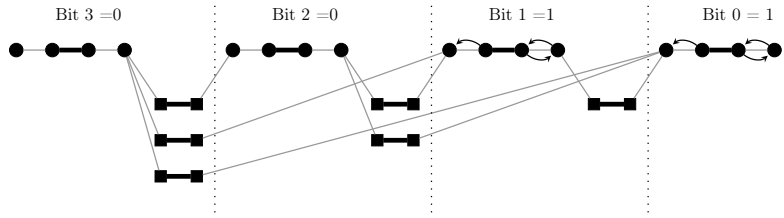
\begin{figure}[!hbt]
 \centering
 \scalebox{0.6}{
\begin{tikzpicture}[xscale = -0.8,yscale=0.8]
	\tikzstyle{legend}=[rectangle, thin,minimum width=2.5cm, minimum height=0.8cm,black]
	\tikzstyle{vertex}=[draw,circle,minimum size=0.3cm,fill=black,inner sep=0pt]
	\tikzstyle{vertex retour}=[draw,rectangle,minimum size=0.25cm,fill=black,inner sep=0pt]
	\tikzstyle{matched edge} = [draw,line width=3pt,-,black]
	\tikzstyle{edge} = [draw,thick,-,black!40]
	\tikzstyle{pointer} = [draw,  thick,->,>=stealth,bend left=40,black,shorten >=1pt]

\foreach \x in {3,4,5,6,9,10,11,12,15,16,17,18,21,22,23,24}
{\node[vertex] (v-\x) at (\x,0.5) {};}

\foreach \from/\to in {3/4,5/6,9/10,11/12,15/16,17/18,21/22,23/24}
{       \draw[edge] (v-\from)--(v-\to);}
\foreach \from/\to in {4/5,10/11,16/17,22/23}
{       \draw[matched edge] (v-\from)to (v-\to);}
\node[vertex retour] (r-32) at (20,-1) {};
\node[vertex retour] (r-31) at (20,-2) {};
\node[vertex retour] (r-30) at (20,-3) {};
\node[vertex retour] (r-21) at (14,-1) {};
\node[vertex retour] (r-20) at (14,-2) {};
\node[vertex retour] (r-10) at ( 8,-1) {};
\node[vertex retour] (l-32) at (19,-1) {};
\node[vertex retour] (l-31) at (19,-2) {};
\node[vertex retour] (l-30) at (19,-3) {};
\node[vertex retour] (l-20) at (13,-2) {};
\node[vertex retour] (l-10) at ( 7,-1) {};
\node[vertex retour] (l-21) at (13,-1) {};
\foreach \x in {30,31,32,21,20,10}
{\draw[matched edge] (r-\x)--(l-\x);}
\draw[edge] (v-9)--(r-10);
\foreach \x in {21,20}
{\draw[edge] (v-15)--(r-\x);}
\foreach \x in {31,30,32}
{\draw[edge] (v-21)--(r-\x);}
\foreach \x in {10,20,30}
{\draw[edge] (v-6)--(l-\x);}
\foreach \x in {21,31}
{\draw[edge] (v-12)--(l-\x);}
\draw[edge] (v-18)--(l-32);
\foreach \x in {6.5,12.5,18.5}
{\draw [loosely dotted,thick] (\x,-4) -- (\x,2) ;}

\draw[pointer] (v-3) to (v-4);
\draw[pointer] (v-4) to (v-3);
\draw[pointer] (v-5) to (v-6);
\draw[pointer] (v-9) to (v-10);
\draw[pointer] (v-10) to (v-9);
\draw[pointer] (v-11) to (v-12);
\node[legend] (b1) at (4,1.5) {Bit 0 = 1};
\node[legend] (b2) at (10,1.5) {Bit 1 =1};
\node[legend] (b3) at (16,1.5) {Bit 2 =0};
\node[legend] (b4) at (22,1.5) {Bit 3 =0};

\end{tikzpicture}
}
 \caption{After turning on the $0$th bit-block,  $G_4$ encodes   0011.}
\label{figure-011}
 \end{figure}

In the case where $i>0$,  for every integer $ j$ from $0$ to $i-1$, edge $(\B{j}{2},\B{j}{3})$ is  in state \emph{On} and  edge $(\B{i}{2},\B{i}{3})$ is in state \emph{Off}.  

More precisely, we can execute the following sequence of moves :

\begin{enumerate}
\item For each integer  $j$, $1 \leq j \leq  i-1$, edge $(\Rr{i}{j},\Rl{i}{j})$ is in  state \emph{Off}. Note that node $\Rr{i}{j}$ (resp.  $Rl{i}{j}$) is adjacent to  one Single node $\B{i}{1}$ (resp.  $\B{j}{4}$). Since  $\B{j}{4}  \leq Lowest \{w\in \Voisin(\Rl{i}{j})  |p_{\Rl{i}{j}} =w\}$,  the hypotheses of  Lemma \ref{lem:offTOon}  are satisfied. Thus from Lemma \ref{lem:offTOon}, we can exhibit an execution to switch edges $(\Rr{i}{j},\Rl{i}{j})$ from state \emph{Off} to state \emph{On}. The configuration shown in Figure \ref{figure-011-1} that corresponds to this step.

 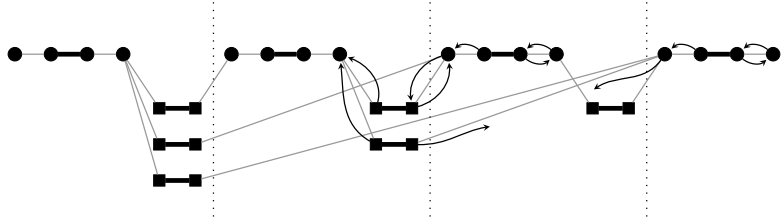
\begin{figure}[!hbt]
 \centering
 \scalebox{0.6}{
\begin{tikzpicture}[xscale = -0.8,yscale=0.8]
	\tikzstyle{legend}=[rectangle, thin,minimum width=2.5cm, minimum height=0.8cm,black]
	\tikzstyle{vertex}=[draw,circle,minimum size=0.3cm,fill=black,inner sep=0pt]
	\tikzstyle{vertex retour}=[draw,rectangle,minimum size=0.25cm,fill=black,inner sep=0pt]
	\tikzstyle{matched edge} = [draw,line width=3pt,-,black]
	\tikzstyle{edge} = [draw,thick,-,black!40]
	\tikzstyle{pointer} = [draw,  thick,->,>=stealth,bend left=40,black,shorten >=1pt]

\foreach \x in {3,4,5,6,9,10,11,12,15,16,17,18,21,22,23,24}
{\node[vertex] (v-\x) at (\x,0.5) {};}

\foreach \from/\to in {3/4,5/6,9/10,11/12,15/16,17/18,21/22,23/24}
{       \draw[edge] (v-\from)--(v-\to);}
\foreach \from/\to in {4/5,10/11,16/17,22/23}
{       \draw[matched edge] (v-\from)to (v-\to);}
\node[vertex retour] (r-32) at (20,-1) {};
\node[vertex retour] (r-31) at (20,-2) {};
\node[vertex retour] (r-30) at (20,-3) {};
\node[vertex retour] (r-21) at (14,-1) {};
\node[vertex retour] (r-20) at (14,-2) {};
\node[vertex retour] (r-10) at ( 8,-1) {};
\node[vertex retour] (l-32) at (19,-1) {};
\node[vertex retour] (l-31) at (19,-2) {};
\node[vertex retour] (l-30) at (19,-3) {};
\node[vertex retour] (l-20) at (13,-2) {};
\node[vertex retour] (l-10) at ( 7,-1) {};
\node[vertex retour] (l-21) at (13,-1) {};
\foreach \x in {30,31,32,21,20,10}
{\draw[matched edge] (r-\x)--(l-\x);}
\draw[edge] (v-9)--(r-10);
\foreach \x in {21,20}
{\draw[edge] (v-15)--(r-\x);}
\foreach \x in {31,30,32}
{\draw[edge] (v-21)--(r-\x);}
\foreach \x in {10,20,30}
{\draw[edge] (v-6)--(l-\x);}
\foreach \x in {21,31}
{\draw[edge] (v-12)--(l-\x);}
\draw[edge] (v-18)--(l-32);
\foreach \x in {6.5,12.5,18.5}
{\draw [loosely dotted,thick] (\x,-4) -- (\x,2) ;}
\draw[pointer] (v-6) to [out=-30,in=170] (8,-0.5);
\draw[pointer] (l-20) to [out=15,in=-175] (10.8,-1.5);
\draw[pointer] (r-20) to [out=-40,in=-165] (v-15);
\draw[pointer] (v-12) to (l-21);
\draw[pointer] (l-21) to (v-12);
\draw[pointer] (r-21) to [out=50,in=-215] (v-15);

\draw[pointer] (v-3) to (v-4);
\draw[pointer] (v-4) to (v-3);
\draw[pointer] (v-5) to (v-6);
\draw[pointer] (v-9) to (v-10);
\draw[pointer] (v-10) to (v-9);
\draw[pointer] (v-11) to (v-12);

\end{tikzpicture}
}
 \caption{After activating the $\carre\,$-nodes of the $3$rd bit-block,
   $G_4$ does not encode any integer.}
\label{figure-011-1}
 \end{figure}

\item Now, for each integer  $j$, $1 \leq j \leq  i-1$,  edge $(\B{j}{2},\B{j}{3})$ is in state \emph{Almost on}.  From Lemma \ref{lem:onTOoff}, (since $Ident(\B{j}{1})< Ident(\B{j}{4})$) an execution to switch edge $(\B{j}{2},\B{j}{3})$ from state \emph{Almost on} to state \emph{Off} is performed. The configuration shown in Figure \ref{figure-011-2} that corresponds to this step.

 \begin{figure}[!hbt]
 \centering
 \scalebox{0.6}{
\begin{tikzpicture}[xscale = -0.8,yscale=0.8]
	\tikzstyle{legend}=[rectangle, thin,minimum width=2.5cm, minimum height=0.8cm,black]
	\tikzstyle{vertex}=[draw,circle,minimum size=0.3cm,fill=black,inner sep=0pt]
	\tikzstyle{vertex retour}=[draw,rectangle,minimum size=0.25cm,fill=black,inner sep=0pt]
	\tikzstyle{matched edge} = [draw,line width=3pt,-,black]
	\tikzstyle{edge} = [draw,thick,-,black!40]
	\tikzstyle{pointer} = [draw,  thick,->,>=stealth,bend left=40,black,shorten >=1pt]

\foreach \x in {3,4,5,6,9,10,11,12,15,16,17,18,21,22,23,24}
{\node[vertex] (v-\x) at (\x,0.5) {};}

\foreach \from/\to in {3/4,5/6,9/10,11/12,15/16,17/18,21/22,23/24}
{       \draw[edge] (v-\from)--(v-\to);}
\foreach \from/\to in {4/5,10/11,16/17,22/23}
{       \draw[matched edge] (v-\from)to (v-\to);}
\node[vertex retour] (r-32) at (20,-1) {};
\node[vertex retour] (r-31) at (20,-2) {};
\node[vertex retour] (r-30) at (20,-3) {};
\node[vertex retour] (r-21) at (14,-1) {};
\node[vertex retour] (r-20) at (14,-2) {};
\node[vertex retour] (r-10) at ( 8,-1) {};
\node[vertex retour] (l-32) at (19,-1) {};
\node[vertex retour] (l-31) at (19,-2) {};
\node[vertex retour] (l-30) at (19,-3) {};
\node[vertex retour] (l-20) at (13,-2) {};
\node[vertex retour] (l-10) at ( 7,-1) {};
\node[vertex retour] (l-21) at (13,-1) {};
\foreach \x in {30,31,32,21,20,10}
{\draw[matched edge] (r-\x)--(l-\x);}
\draw[edge] (v-9)--(r-10);
\foreach \x in {21,20}
{\draw[edge] (v-15)--(r-\x);}
\foreach \x in {31,30,32}
{\draw[edge] (v-21)--(r-\x);}
\foreach \x in {10,20,30}
{\draw[edge] (v-6)--(l-\x);}
\foreach \x in {21,31}
{\draw[edge] (v-12)--(l-\x);}
\draw[edge] (v-18)--(l-32);
\foreach \x in {6.5,12.5,18.5}
{\draw [loosely dotted,thick] (\x,-4) -- (\x,2) ;}
\draw[pointer] (v-6) to [out=-30,in=170] (8,-0.5);
\draw[pointer] (l-20) to [out=15,in=-175] (10.8,-1.5);
\draw[pointer] (r-20) to [out=-40,in=-165] (v-15);
\draw[pointer] (v-12) to (l-21);
\draw[pointer] (l-21) to (v-12);
\draw[pointer] (r-21) to [out=50,in=-215] (v-15);
\end{tikzpicture}
}
 \caption{Starting to turn off the $0$th and $1$st bit-blocks.}
\label{figure-011-2}
 \end{figure}
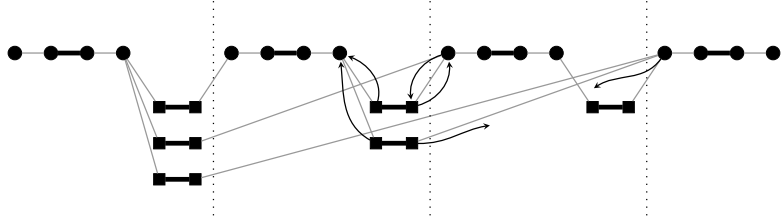

\item  Edge $(\B{i}{2},\B{i}{3})$ is still in state \emph{off}. Using the same argument of step (1),  from Lemma \ref{lem:offTOon}, we can exhibit an execution to switch edges $(\B{i}{2},\B{i}{3})$ from state \emph{Off} to state \emph{On}.

 \begin{figure}[!hbt]
 \centering
 \scalebox{0.6}{
\begin{tikzpicture}[xscale = -0.8,yscale=0.8]
	\tikzstyle{legend}=[rectangle, thin,minimum width=2.5cm, minimum height=0.8cm,black]
	\tikzstyle{vertex}=[draw,circle,minimum size=0.3cm,fill=black,inner sep=0pt]
	\tikzstyle{vertex retour}=[draw,rectangle,minimum size=0.25cm,fill=black,inner sep=0pt]
	\tikzstyle{matched edge} = [draw,line width=3pt,-,black]
	\tikzstyle{edge} = [draw,thick,-,black!40]
	\tikzstyle{pointer} = [draw, thick,->,>=stealth,bend left=40,black,shorten >=1pt]

\foreach \x in {3,4,5,6,9,10,11,12,15,16,17,18,21,22,23,24}
{\node[vertex] (v-\x) at (\x,0.5) {};}

\foreach \from/\to in {3/4,5/6,9/10,11/12,15/16,17/18,21/22,23/24}
{       \draw[edge] (v-\from)--(v-\to);}
\foreach \from/\to in {4/5,10/11,16/17,22/23}
{       \draw[matched edge] (v-\from)to (v-\to);}
\node[vertex retour] (r-32) at (20,-1) {};
\node[vertex retour] (r-31) at (20,-2) {};
\node[vertex retour] (r-30) at (20,-3) {};
\node[vertex retour] (r-21) at (14,-1) {};
\node[vertex retour] (r-20) at (14,-2) {};
\node[vertex retour] (r-10) at ( 8,-1) {};
\node[vertex retour] (l-32) at (19,-1) {};
\node[vertex retour] (l-31) at (19,-2) {};
\node[vertex retour] (l-30) at (19,-3) {};
\node[vertex retour] (l-20) at (13,-2) {};
\node[vertex retour] (l-10) at ( 7,-1) {};
\node[vertex retour] (l-21) at (13,-1) {};
\foreach \x in {30,31,32,21,20,10}
{\draw[matched edge] (r-\x)--(l-\x);}
\draw[edge] (v-9)--(r-10);
\foreach \x in {21,20}
{\draw[edge] (v-15)--(r-\x);}
\foreach \x in {31,30,32}
{\draw[edge] (v-21)--(r-\x);}
\foreach \x in {10,20,30}
{\draw[edge] (v-6)--(l-\x);}
\foreach \x in {21,31}
{\draw[edge] (v-12)--(l-\x);}
\draw[edge] (v-18)--(l-32);
\foreach \x in {6.5,12.5,18.5}
{\draw [loosely dotted,thick] (\x,-4) -- (\x,2) ;}
\draw[pointer] (v-6) to [out=-30,in=170] (8,-0.5);
\draw[pointer] (l-20) to [out=15,in=-175] (10.8,-1.5);
\draw[pointer] (r-20) to [out=-40,in=-165] (v-15);
\draw[pointer] (v-12) to (l-21);
\draw[pointer] (l-21) to (v-12);
\draw[pointer] (r-21) to [out=50,in=-215] (v-15);
\draw[pointer] (v-15) to (v-16);
\draw[pointer] (v-16) to (v-15);
\draw[pointer] (v-17) to (v-18);

\end{tikzpicture}
}
 \caption{Starting to turn on the $3$rd  bit-block.}
\label{figure-011-3}
 \end{figure}
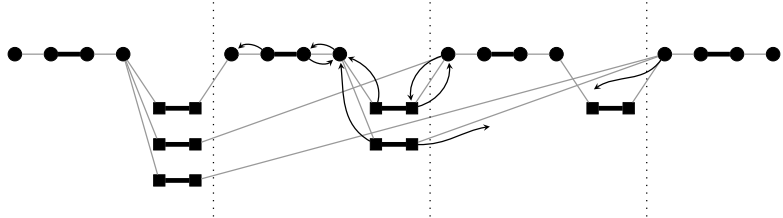

\item Now, for each integer $j$, $1 \leq j \leq  i-1$, edge $(\Rr{i}{j},\Rl{i}{j})$ is now in state \emph{Almost on}. From Lemma \ref{lem:onTOoff}, an execution to switch edge $(\Rr{i}{j},\Rl{i}{j})$ from state \emph{Almost on} to state \emph{Off}.

 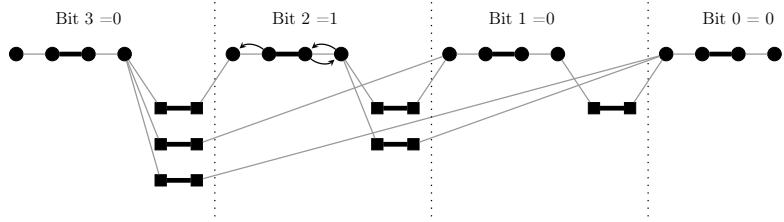
\begin{figure}[!hbt]
 \centering
 \scalebox{0.6}{

\begin{tikzpicture}[xscale = -0.8,yscale=0.8]
	\tikzstyle{legend}=[rectangle, thin,minimum width=2.5cm, minimum height=0.8cm,black]
	\tikzstyle{vertex}=[draw,circle,minimum size=0.3cm,fill=black,inner sep=0pt]
	\tikzstyle{vertex retour}=[draw,rectangle,minimum size=0.25cm,fill=black,inner sep=0pt]
	\tikzstyle{matched edge} = [draw,line width=3pt,-,black]
	\tikzstyle{edge} = [draw,thick,-,black!40]
	\tikzstyle{pointer} = [draw, thick,->,>=stealth,bend left=40,black,shorten >=1pt]
\foreach \x in {3,4,5,6,9,10,11,12,15,16,17,18,21,22,23,24}
{\node[vertex] (v-\x) at (\x,0.5) {};}

\foreach \from/\to in {3/4,5/6,9/10,11/12,15/16,17/18,21/22,23/24}
{       \draw[edge] (v-\from)--(v-\to);}
\foreach \from/\to in {4/5,10/11,16/17,22/23}
{       \draw[matched edge] (v-\from)to (v-\to);}
\node[vertex retour] (r-32) at (20,-1) {};
\node[vertex retour] (r-31) at (20,-2) {};
\node[vertex retour] (r-30) at (20,-3) {};
\node[vertex retour] (r-21) at (14,-1) {};
\node[vertex retour] (r-20) at (14,-2) {};
\node[vertex retour] (r-10) at ( 8,-1) {};
\node[vertex retour] (l-32) at (19,-1) {};
\node[vertex retour] (l-31) at (19,-2) {};
\node[vertex retour] (l-30) at (19,-3) {};
\node[vertex retour] (l-20) at (13,-2) {};
\node[vertex retour] (l-10) at ( 7,-1) {};
\node[vertex retour] (l-21) at (13,-1) {};
\foreach \x in {30,31,32,21,20,10}
{\draw[matched edge] (r-\x)--(l-\x);}
\draw[edge] (v-9)--(r-10);
\foreach \x in {21,20}
{\draw[edge] (v-15)--(r-\x);}
\foreach \x in {31,30,32}
{\draw[edge] (v-21)--(r-\x);}
\foreach \x in {10,20,30}
{\draw[edge] (v-6)--(l-\x);}
\foreach \x in {21,31}
{\draw[edge] (v-12)--(l-\x);}
\draw[edge] (v-18)--(l-32);
\foreach \x in {6.5,12.5,...,18.5}
{\draw [loosely dotted,thick] (\x,-4) -- (\x,2) ;}

\draw[pointer] (v-15) to (v-16);
\draw[pointer] (v-16) to (v-15);
\draw[pointer] (v-17) to (v-18);

\node[legend] (b1) at (4,1.5) {Bit 0 = 0};
\node[legend] (b2) at (10,1.5) {Bit 1 =0};
\node[legend] (b3) at (16,1.5) {Bit 2 =1};
\node[legend] (b4) at (22,1.5) {Bit 3 =0};

\end{tikzpicture}
}
 \caption{Ending  to turn off the 0th and 1st bit-blocks and to turn on  the $3$rd bit-block. $G_4$ encodes  0100.}
\label{figure-100}
 \end{figure}
\end{enumerate}
  
At the end of this execution, the configuration still verifies the two conditions, and the $i-1$ first bits of $\omega$ are set to $0$ and the $i^{th}$ to $1$. So we obtain a $(\omega+1)$-configuration.  
\end{proof}

From now, we can construct an instance from which an execution having $\Omega(2^{\sqrt n})$ moves  can be built.

\begin{corollary}\label{cor:execution:expo}
Algorithm $\mathcal{M}^+$ can stabilize after at most $\Omega(2^{\sqrt n})$ moves under the central daemon.
\end{corollary}

\begin{proof}
To prove the corollary, we can exhibit an execution of $\Omega(2^{\sqrt n})$ moves.  Let $N$ be an integer. The initial configuration is a $0$-configuration in graph $G_N$.

We can build an execution that contains all the $\omega$-configurations for every value $\omega$, $1\leq \omega \leq2^N$.  By applying Theorem~\ref{th:execution:+1}, this execution can be split into $ 2^N$ parts corresponding to the execution from $\omega$-configuration to $(\omega+1)$-configuration, for $1\leq \omega \leq2^N$. Thus, this execution has $O(2^{N})$ configurations. Since graph $G_N$ has $O(N^2)$ vertices, this execution has $O(2^{\sqrt n})$ configurations and the corollary holds.
\end{proof}

\end{document}